\documentclass{www2013-accepted}
 \pdfoutput=1
 
\usepackage{verbatim}
\usepackage{amsmath}
\usepackage{amssymb}
\usepackage{url}
\usepackage{array}
\usepackage{booktabs}
\usepackage{afterpage}
\usepackage[dvipsnames]{color}
\usepackage{mfirstuc}
\usepackage{subfigure}
\usepackage{times}
\makeatletter
\newif\if@restonecol
\makeatother

\usepackage[lined,algonl,boxed]{algorithm2e}
\usepackage{epsfig}
\usepackage{multirow}
\usepackage{rotating}
\usepackage{enumitem}

\renewcommand{\Pr}{{\text{Pr}}}
\newtheorem{theorem}{Theorem}[section]

\newtheorem{lemma}[theorem]{Lemma}

\newtheorem{program}[theorem]{Program}
\newtheorem{proposition}[theorem]{Proposition}

\def\comp{\overline}

\def\inj{\text{inj}}
\def\aut{\text{aut}}
\def\ext{\text{ext}}

\def\tinj{t_{\text{inj}}}

\def\Gnp{G_{n,p}}
\def\Pr{\text{Pr}}

\newcommand{\omt}[1]{}
\newcommand{\xhdr}[1]{\paragraph*{\bf #1}}

\title{Subgraph Frequencies: Mapping the Empirical and Extremal Geography of Large Graph Collections}
\numberofauthors{3} 
\author{
\alignauthor
Johan Ugander\\
       \affaddr{Cornell University}\\
       \affaddr{Ithaca, NY}\\
       \email{jhu5@cornell.edu}
\alignauthor
Lars Backstrom\\ 
       \affaddr{Facebook}\\
       \affaddr{Menlo Park, CA}\\
       \email{lars@fb.com}
\alignauthor
Jon Kleinberg\\ 
       \affaddr{Cornell University}\\
       \affaddr{Ithaca, NY}\\
       \email{kleinber@cs.cornell.edu}
}
\date{}

\begin{document}
\maketitle
\begin{abstract}
A growing set of on-line applications are generating data that 
can be viewed as very large collections of small, dense social 
graphs --- these range from sets of social groups, events, 
or collaboration projects to the vast collection of graph neighborhoods
in large social networks. A natural question is how to usefully define
a domain-independent `coordinate system' for such a collection
of graphs, so that the set of possible structures can be compactly
represented and understood within a common space. 
In this work, we draw on the theory of graph homomorphisms to 
formulate and analyze such a representation, based on computing
the frequencies of small induced subgraphs within each graph. 
We find that the space of subgraph frequencies is governed both
by its combinatorial properties --- based on extremal results that 
constrain all graphs --- as well as by its empirical properties --- 
manifested in the way that real social graphs appear to lie near 
a simple one-dimensional curve through this space.

We develop flexible frameworks for studying each of these 
aspects. For capturing empirical properties, we characterize
a simple stochastic generative model, a single-parameter extension 
of Erd\H{o}s-R\'{e}nyi random graphs,
whose stationary distribution over subgraphs closely tracks 
the one-dimensional concentration of the real social graph families. For 
the extremal properties, we develop a tractable linear program 
for bounding the feasible space of subgraph frequencies 
by harnessing a toolkit of known extremal graph theory. 
Together, these two complementary frameworks shed light
on a fundamental question pertaining to social graphs: 
what properties of social graphs are `social' properties and
what properties are `graph' properties?

We conclude with a brief demonstration of how the coordinate
system we examine can also be used to 
perform classification tasks, distinguishing between 
structures arising from different types of social graphs.
\end{abstract}

\vspace{0.05in}
\noindent
{\bf Categories and Subject Descriptors:}
H.2.8 [{\bf Database Management}]: Database applications---{\em Data mining}

\noindent
{\bf Keywords:}
Social Networks, Triadic Closure, Induced Subgraphs, Subgraph Census, Graph Homomorphisms.

\section{Introduction}

The standard approach to modeling a large on-line social network is to
treat it as a single graph with an enormous number of nodes and a sparse
pattern of connections.
Increasingly, however, many of the key problems
encountered in managing an on-line social network involve working with
large collections of small, dense graphs contained within the network.

On Facebook, for example, the set of people belonging to a group or
attending an event determines such a graph, and considering 
the set of all groups or all events leads to a very large number
of such graphs.
On any social network, the network neighborhood of each individual --- consisting of
his or her friends and the links among them --- is also generally
a small dense graph with a rich structure, on a few hundred nodes or 
fewer \cite{UKBM12}.
If we consider the neighborhood of each user as defining a distinct graph,
we again obtain an enormous collection of graphs.
Indeed, this view of a large underlying
social network in terms of its overlapping node neighborhoods
suggests a potentially valuable perspective on the analysis of the network:
rather than thinking of Facebook, for example, as a single billion-node
network, with a global structure that quickly becomes incomprehensible,
we argue that it can be useful to think of it
as the superposition of a billion small dense graphs --- the network
neighborhoods, one centered at each user, and each
accessible to a closer and more tractable investigation.

Nor is this view limited to a site such as Facebook; 
one can find collections of small dense graphs in the interactions
within a set of discussion forums \cite{welser},
within a set of collaborative on-line projects \cite{voss-wikipedia}, 
and in a range of other settings.

Our focus in the present work is on a fundamental global question
about these types of graph collections:
given a large set of small dense graphs, can we study 
this set by defining a meaningful `coordinate system' on it,
so that the graphs it contains can be represented and
understood within a common space?
With such a coordinate system providing a general-purpose framework for
analysis, additional questions become possible.
For example, 
when considering collections of a billion or more social graphs, it
may seem as though almost any graph is possible; 
is that the case,
or are there underlying properties guiding the observed structures? 
And how do these properties relate to more fundamental 
combinatorial constraints deriving from the extremal limits that
govern all graphs?
As a further example,
we can ask how different graph collections compare to one another;
do network neighborhoods differ in some systematic way, for instance, 
from social graphs induced by other contexts, such as the graphs
implicit in social groups, organized events, or other arrangements?

\paragraph*{\bf The Present Work}
In this paper we develop and analyze such a representation,
drawing on the theory of {\em graph homomorphisms}.
Roughly speaking, the coordinate system we examine
begins by describing a graph by the frequencies with
which all possible small subgraphs occur within it.
More precisely, we choose a small number $k$ (e.g. $k = 3$ or $4$);
then, for each graph $G$ in a collection, 
we create a vector with a coordinate for each 
distinct $k$-node subgraph $H$, specifying the fraction of $k$-tuples
of nodes in $G$ that induce a copy of $H$ (in other words, the frequency of 
$H$ as an induced subgraph of $G$). 
For $k=3$, this description corresponds
to what is sometimes referred to as the {\em triad census}
\cite{Davis1967,faust2007,faust2010,wasserman-faust}. 
The literature on {\em frequent subgraph mining} 
\cite{Inokuchi2000,kuramochi2004,yan2002}, and {\em motif counting}
\cite{Milo2002} is also is closely related, but focuses on connected subgraphs.

With each graph in the collection mapped to such a vector, we can ask
how the full collection of graphs fills out this space of
subgraph frequencies.
This turns out to be a subtle issue, because the arrangement of 
the graphs in this space is governed by two distinct sets of
effects: extremal combinatorial constraints showing that certain 
combinations of subgraph frequencies are genuinely impossible;
and empirical properties, which reveal that the bulk of the 
graphs tend to lie close to a simple one-dimensional curve
through the space.
We formulate results on both these types of properties, in the former
case building on an expanding body of combinatorial theory
\cite{borgs,lovasz} for bounding the 
frequencies at which different types of subgraphs can occur
in a larger ambient graph.

The fact that the space of subgraph frequencies is constrained
in these multiple ways also allows us to concretely address the 
following type of question:
When we see that human social networks do not exhibit a certain type
of structure, is that because such a structure is mathematically
impossible, or simply because human beings do not create it when
they form social connections?  In other words, what is a property
of graphs and what is a property of people?
Although this question is implicit in many studies of social networks,
it is hard to separate the two effects without a formal framework
such as we have here.

Indeed, our framework offers a direct contribution to
one of the most well-known observations about social graphs: 
the tendency of social relationships to close triangles, and the 
relative infrequency of what is 
sometimes called the `forbidden triad': 
three people with two social relationships between them, 
but one absent relationship \cite{rapoport-triadic}.
There are many sociological theories for why one would expect this
subgraph to be underrepresented in empirical social networks 
\cite{granovetter-weak-ties}.
Our framework 
shows that the frequency of this `forbidden triad' 
has a non-trivial upper bound in not just social graphs, 
but in all graphs.
Harnessing our framework more
generally, we are in fact able to show that {\em any} $k$ node 
subgraph that is not a complete or empty subgraph
has a frequency that is bounded away from one. 
Thus, there is an extent to which almost all subgraphs are
mathematically `forbidden' from occurring beyond
 a certain frequency.

We aim to separate these mathematical limits of graphs from the
complementary empirical properties of real social graphs. The fact that
real graph collections have a roughly one-dimensional 
structure in our coordinate system leads directly to our first main question: 
is it possible to succinctly characterize the underlying backbone for this 
one-dimensional structure, and can we use such a 
characterization to usefully describe
graphs within our coordinate system in terms of their
deviation from this backbone?

The subgraph frequencies of the standard
Erd\H{o}s-R\'{e}nyi random graph \cite{bollobas-rand-graphs-book} 
$G_{n,p}$ produce a one-dimensional curve
(parametrized by $p$) that weakly approximates the layout of the real
graphs in the space, but the curve arising from this random graph model
systematically deviates from the real graphs in that the random graph
contains fewer triangles and more triangle-free subgraphs.
This observation is consistent with the sociological principle
of {\em triadic closure} --- 
that triangles tend to form in social networks. 
As a means of closing this deviation from $G_{n,p}$, 
we develop a tractable stochastic model of graph generation
with a single additional parameter, determining the relative rates of
arbitrary edge formation and triangle-closing edge formation.
The model exhibits rich behaviors, and for appropriately chosen 
settings of its single parameter, 
it produce remarkably close 
agreement 
with the
subgraph frequencies observed in real data for the suite of
all possible 3-node and 4-node subgraphs.

Finally, we use this representation to study how different collections
of graphs may differ from one another.
This arises as a question of basic interest in the analysis of large
social media platforms, where users continuously manage multiple audiences 
\cite{BBKLR11} ---
ranging from their set of friends, to the members of a groups they've joined,
to the attendees of events and beyond.
Do these audiences differ from each other at a structural level, and
if so what are the distinguishing characteristics?
Using Facebook data, we identify structural differences between
the graphs induced on network neighborhoods, groups, and events.
The underlying basis for these differences suggests corresponding
distinctions in each user's reaction to these different audiences with whom
they interact.

\section{Data description}
\label{sec:data}

Throughout our presentation, we analyze several collections of graphs
collected from Facebook's social network. The collections we study
are all induced graphs from the Facebook friendship graph, which
records friendship connections as undirected edges between users, 
and thus all our induced graphs are also undirected. 
The framework we characterize in this work would naturally extend to 
provide insights about directed graphs, an extension we do not discuss.
We do not include edges formed by Facebook `subscriptions' in our study, 
nor do we include Facebook `pages' or connections from users 
to such pages. 
All Facebook social graph data was analyzed in an anonymous, aggregated form.

For this work, we extracted three different collections of graphs, around which we
organize our discussion: 
\begin{itemize}[leftmargin=10pt]
\itemsep1pt 
\parskip2pt 
\parsep2pt
\labelsep4pt
\item {\bf Neighborhoods}: Graphs induced by the friends 
of a single Facebook user {\it ego} and the friendship 
connections among these individuals (excluding the ego).
\item {\bf Groups}: Graphs induced by the members of a 
`Facebook group', a Facebook feature for organizing focused 
conversations between a small or moderate-sized set of users.
\item {\bf Events}: Graphs induced by the confirmed attendees of 
`Facebook events', a Facebook feature  for coordinating invitations
to calendar events. Users can response `Yes', `No', and `Maybe' to 
such invitations, and we consider only users who respond `Yes'.
\end{itemize}

The neighborhood and groups collections were 
assembled in October 2012 based on monthly active user egos 
and current groups, while the events data was collected from 
all events during 2010 and 2011. For event graphs, only friendship 
edges formed prior to the date of the event were used. 
Subgraph frequencies
for four-node subgraphs were computed 
by sampling 11,000 induced subgraphs
uniformly with replacement, providing sufficiently precise
frequencies without enumeration.
The graph collections were targeted at a variety of different 
graph sizes, as will be discussed in the text.

\begin{figure*}[t]
\begin{center}
\vspace*{-0.225in}
\includegraphics[width=0.9\linewidth]{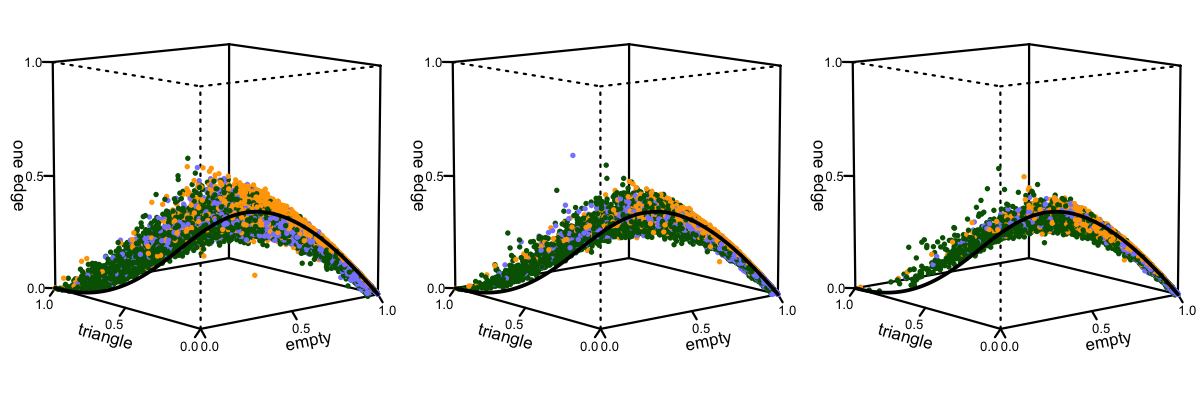} 
\vspace*{-0.375in}
\caption{Subgraph frequencies for three node subgraphs for 
graphs of size 50, 100, and 200 (left to right). The
neighborhoods are orange, groups are green, and events are lavender.
The black curves illustrate $\Gnp$ as a function of $p$.
\vspace*{-0.3in}
}
\label{f:3d}
\end{center}
\end{figure*}

\section{Subgraph space}
\label{sec:triads}

In this section, we study the space of subgraph frequencies
that form the basis of our coordinate system, and the
one-dimensional concentration of empirical graphs within
this coordinate system. We derive a model capable of accurately 
identifying the backbone of this empirical concentration using
only the basic principle of {\em triadic closure},
showing how the subgraph frequencies of empirical social graphs
are seemingly restricted to the vicinity of a simple one-dimensional
structure.

Formally, the subgraph frequency of a $k$-node graph $F$ in an 
$n$-node graph $G$ (where $k \leq n$) is the probability that
a random $k$-node subset of $G$ induces a copy of $F$.
It is clear that for any integer $k$, the subgraph frequencies
of all the $k$-node graphs sum to one, constraining the vector
of frequencies to an appropriately dimensioned simplex. 
In the case of $k=3$, this vector is simply the relative frequency
of induced three-node subgraphs restricted to the 4-simplex;
there are just four such subgraphs,
with zero, one, two, and three edges respectively. 
When considering
the frequency of larger subgraphs, the dimension of the simplex
grows very quickly, and already for $k=4$, 
the space of four-node subgraph frequencies lives in an 11-simplex.

\xhdr{Empirical distribution}
In Figure \ref{f:3d}, the three-node subgraph frequencies of
50-node, 100-node, and 200-node graph collections are shown, with each
subplot showing a balanced mixture of 17,000 
neighborhood, group and event graphs -- the three
collections discussed in Section \ref{sec:data}, totaling 51,000 graphs
at each size. Because
these frequency vectors are constrained to the 4-simplex, their distribution
can be visualized in $\mathbb R^3$ with three of the frequencies
as axes. 

Notice that these graph collections, induced from disparate contexts, 
all occupy a sharply concentrated subregion of the unit simplex. The
points in the space have been represented simply as an unordered
scatterplot, and two striking phenomena already stand out:
first, the particular concentrated structure within 
the simplex that the points follow; and second, the fact that
we can already discern a non-uniform distribution
of the three contexts (neighborhoods, groups and events) within the space ---
that is, the different contexts can already be seen to have
different structural loci. Notice also that as the sizes of the 
graphs increases -- from 50 to 100 to 200 -- the distribution
appears to sharpen around the one-dimensional backbone. 
The vast number of graphs that we are able
to consider by studying Facebook data is here illuminating 
a structure that is 
simply not discernible in previous examinations of subgraph
frequencies \cite{faust2010}, since no analysis has previously considered a collection
near this scale.

The imagery of Figure \ref{f:3d} directly motivates our work,
by visually framing the essence of our investigation: what facets of this
curious structure derive from our graphs being social graphs, and
what facets are simply universal properties of all graphs?
We will find, in particular, that parts of the space of subgraph
frequencies are in fact inaccessible to graphs for purely
combinatorial reasons --- it is mathematically impossible for one of the points
in the scatterplot to occupy these parts of the space.
But there are other parts of the space that are mathematically possible;
it is simply that no real social graphs appear to be located within them.
Intuitively, then, we are looking at a population density within an
ambient space (the Facebook graphs within the space of 
subgraph frequencies), and we would like to understand both 
the 
geography of the inhabited terrain (what are the
properties of the areas where the population has in fact settled?)
and also 
the properties of the boundaries of the space as a whole
(where, in principle, would it be possible for the population to settle?).

Also in Figure \ref{f:3d}, we plot the curve for the
frequencies for 3 node subgraphs in $\Gnp$ as a function of $p$. 
The curves are given simply by the probability of obtaining the desired
number of edges in a three node graph, $((1-p)^3,3p(1-p)^2,3p^2(1-p),p^3)$. 
This curve closely tracks the empirical density through the space,
with a single notable discrepancy: 
the real world graphs systemically contain more 
triangles when compared to $\Gnp$ at the 
same edge density. 
We emphasize that it is not a priori clear why $\Gnp$ would
at all be a good model of subgraph frequencies in modestly-sized
dense social graphs such as the neighborhoods, groups, and events
that we have here;
we believe the fact that it tracks the data 
with any fidelity at all is an interesting issue
for future work.
Beyond $\Gnp$, in the following subsection, we present a stochastic model
of edge formation and deletion on graphs specifically designed
to close the remaining discrepancy. As such, our model provides a means of
accurately characterizing the backbone of
subgraph frequencies for social graphs. 

\begin{figure*}[t]
\begin{center}
\vspace*{-0.175in}
\includegraphics[width=0.75\linewidth]{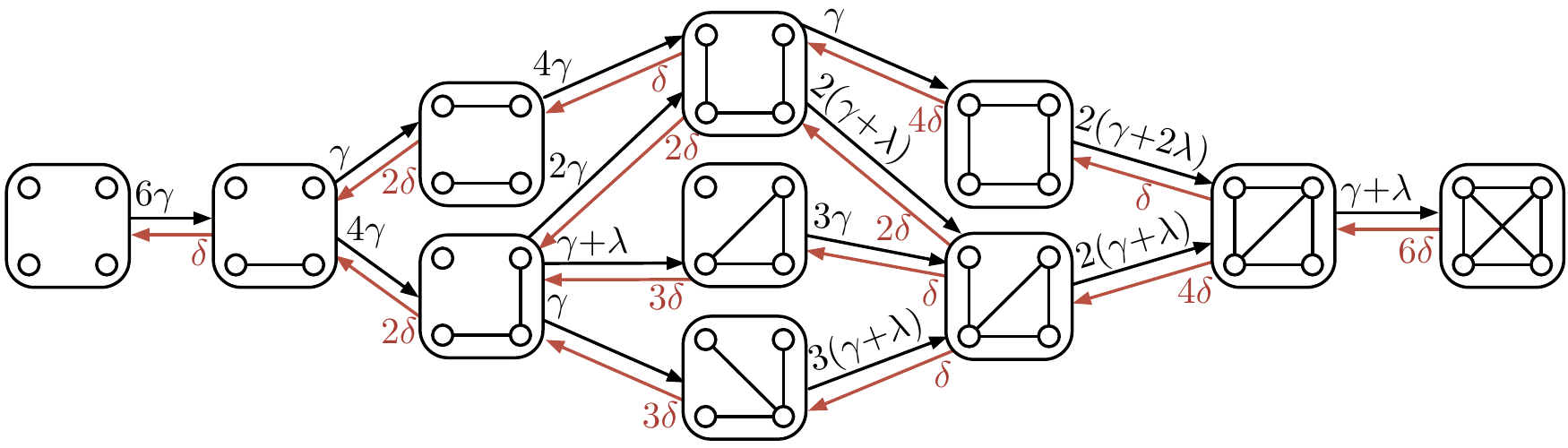} 
\vspace*{-0.15in}
\caption{
The state transitions diagram for our stochastic graph model
with $k=4$, where 
$\gamma$ is the arbitrary edge formation rate,
$\lambda$ is the triadic closure formation rate,
and $\delta$ is the edge elimination rate.
\vspace*{-0.3in}
}
\label{f:transitions}
\end{center}
\end{figure*}

\xhdr{Stochastic model of edge formation}
The classic Erd\H{o}s-R\'{e}nyi model of random graphs, $\Gnp$,
produces a distribution over $n$-node undirected graphs defined by a simple
parameter $p$, the probability of each edge independently appearing 
in the graph. We now introduce and analyze a related 
random graph model, the {\em Edge Formation Random Walk}, defined
as a random walk over the space of all unlabeled $n$-node graphs. 
In its simplest form,
this model is closely related to $\Gnp$, and will we show via detailed balance
that the distribution defined by $\Gnp$ on $n$-node graphs
is precisely the stationary distribution of this simplest version
of the random walk on the space of $n$-node graphs. 
We first describe this basic version of the model; we then 
add a component to the model that captures a triadic closure process,
which produces a close fit to the properties we observe in real graphs.

Let $\mathcal G_n$ be the space of all unlabeled $n$-node graphs, and 
let $X(t)$ be the following continuous 
time Markov chain on the state space $\mathcal G_n$.
The transition rates between the graphs in $\mathcal G_n$ are defined by
random additions and deletions of edges, with all edges having a uniform formation
rate $\gamma>0$ and a uniform deletion rate $\delta>0$. 
Thus the single parameter $\nu = \gamma / \delta$, the effective formation rate of edges,
completely characterizes the process.
Notice that this process is clearly irreducible, 
since it is possible to transition between any two graphs via edge additions
and deletions. 

Since $X(t)$ is irreducible, it possesses a unique stationary distribution.
The stationary distribution of an irreducible continuous time Markov chain 
can be found as the unique stable fixed point of the 
linear dynamical system $X'(t) = Q_n(\nu)X(t)$ that describes the diffusion
of probability mass during a random walk on $n$-node graphs,
where $Q_n(\nu)$ is the generator matrix
with transition rates $q_{ij}$ and $q_{ii} = - \sum_{j \ne i} q_{ji}$, all depending
only on $\nu$. The stationary distribution $\pi_n$ then satisfies $Q_n(\nu)\pi_n = 0$.

The following proposition shows the clear relationship between the stationary distribution
of this simplest random walk and the frequencies of $\Gnp$.
\begin{proposition}
\label{prop:balance}
The probabilities assigned to (unlabeled) graphs by $\Gnp$
satisfy the detailed balance condition
for the Edge Formation Random Walk with edge
formation rate $\nu = p / 1 - p$, and thus characterizes
the stationary distribution. 
\end{proposition}
\begin{proof}
We first describe an equivalent Markov chain based on labeled graphs:
there is a state for each labeled $n$-node graph;
the transition rate $q_{ij}$ from a labelled graph $G_i$ to a labelled graph $G_j$
is $q_{ij}=\gamma$ if $G_j$ can be obtained from $G_i$ by adding an edge;
and 
$q_{ij}=\delta$ if $G_j$ can be obtained from $G_i$ by removing an edge. All 
other transition rates are zero.
We call this new chain the {\em labeled chain}, and the original
chain the {\em unlabeled chain.}

Now, suppose there is a transition from unlabeled graph $H_a$ to unlabeled
graph $H_b$ in the unlabeled chain, with transition probability
$k \gamma$.
This means that there are $k$ ways to add an edge to a labeled copy of
$H_a$ to produce a graph isomorphic to $H_b$.
Now, let $G_i$ be any graph in the labeled chain that is isomorphic to $H_a$.
In the labeled chain, there are $k$ transitions out of $G_i$ leading
to a graph isomorphic to $H_b$, and each of these has probability 
$\gamma$.  Thus, with probability $k \gamma$, a transition out of $G_i$
leads to a graph isomorphic to $H_b$.
A strictly analogous argument can be made for edge deletions, rather
than edge additions.

This argument shows that the following describes a Markov chain 
equivalent to the original unlabeled chain: we draw a sequence of
labeled graphs from the labeled chain, and we output the isomorphism
classes of these labeled graphs.  Hence, to compute the stationary
distribution of the original unlabeled chain, which is what we seek,
we can compute the stationary distribution of the labeled chain and
then sum stationary probabilities in the labeled chain
over the isomorphism classes of labeled graphs.

It thus suffices to verify the detailed balance condition
for the distribution on the labeled chain that assigns probability
$p^{|E(G_i)|}(1-p)^{{n \choose 2} - |E(G_i)|}$ to each labeled graph $G_i$.
Since every transition of the labeled walk
occurs between two labeled graphs $G_i$ and $G_j$, 
with $|E(G_i)| = |E(G_j)| +1$, 
the only non-trivial detailed balance equations are of the form:
\begin{eqnarray*}
q_{ij} \Pr[X(t) = G_i] &=& q_{ji} \Pr[X(t) =G_j]\\
 \Pr[X(t) = G_i] &=& \nu \Pr[X(t) =G_j] \\
\Pr[X(t) =G_i] &=& \frac{p}{1-p} \Pr[X(t) =G_j].
\end{eqnarray*}
Since the probability assigned to the labeled graph $G_i$ by $\Gnp$
is simply $p^{|E(G_i)|}(1-p)^{{n \choose 2} - |E(G_i)|}$,
detailed balance is clearly satisfied.
\end{proof}

\xhdr{Incorporating triadic closure}
The above modeling framework provides a simple analog of $\Gnp$ that 
notably exposes itself to subtle adjustments. By simply adjusting the 
transition rates between select graphs, this framework makes it possible
to model random graphs where certain types of edge
formations or deletions have irregular probabilities of occurring, simply via
small perturbations away from the classic $\Gnp$ model. 
Using this principle, we now characterize a 
random graph model that differs from $\Gnp$ by a single parameter,
$\lambda$, the rate at which 3-node paths in the graph tend to form triangles.
We call this model the {\em Edge Formation Random Walk with Triadic Closure}.

Again let $\mathcal G_n$ be the space of all unlabeled $n$-node graphs, and 
let $Y(t)$ be a continuous time Markov chain on the state space $\mathcal G_n$.
As with the ordinary Edge Formation Random Walk,
let edges have a uniform formation
rate $\gamma > 0$ and a uniform deletion rate $\delta > 0$, but now also 
add a triadic closure formation rate $\lambda \ge 0$ for every 
3-node path that a transition would close.
The process is still clearly irreducible, and the stationary distribution obeys
the stationary conditions $Q_n(\nu,\lambda) \pi_n =0$, where the generator
matrix $Q_n$ now also depends on $\lambda$. We can express the
stationary distribution directly in the parameters as 
$\pi_n(\nu,\lambda) = \{\pi: Q_n(\nu,\lambda)\pi=0\}$.
For $\lambda=0$ the model reduces to the ordinary Edge Formation
Random Walk.   

The state transitions of this random graph model are easy to construct
for $n=3$ and $n=4$, and transitions for the case of $n=4$
are shown in Figure \ref{f:transitions}. Proposition \ref{prop:balance}
above tells us
that for $\lambda=0$, the stationary distribution of a random walk
on this state space is given by the graph
frequencies of $\Gnp$. As we increase $\lambda$ away from zero,
we should therefore expect to see a stationary distribution that
departs from $\Gnp$ precisely by observing
more graphs with triangles and less graphs with open triangles. 

The framework of our Edge Formation Random Walk makes it possible
to model triadic closure precisely;
in this sense the model forms an interesting contrast with 
other models of triangle-closing in graphs that
are very challenging to analyze (e.g. \cite{chatterjee2011,jin-girvan-newman,leskovec2005,park2005,strauss1986}).
We will now show how 
the addition of this single parameter makes it possible to describe
the subgraph frequencies of empirical social graphs with remarkable
accuracy.

\xhdr{Fitting subgraph frequencies}

The stationary distribution of an Edge Formation Random Walk model
describes the frequency of different graphs, while the coordinate system
we are developing focuses on the frequency of $k$-node subgraphs within
$n$-node graphs. For $\Gnp$ these two questions are in fact the same,
since the distribution of random induced $k$-node subgraphs of $\Gnp$ 
is simply $G_{k,p}$. When we introduce $\lambda>0$, however, our model
departs from this symmetry, and the stationary probabilities in a
random walk on $k$ node graphs is no longer precisely the frequencies of
induced $k$-node subgraphs in a single $n$-node graph. 

But if we view this as a model for the frequency of small graphs
as objects in themselves, rather than as subgraphs of a larger 
ambient graph, the model provides a 
highly tractable parameterization 
that we can use to approximate the structure of subgraph frequencies 
observed in our families of larger graphs.
In doing so, we aim to fit $\pi_k(\nu(p,\lambda),\lambda)$
as a function of $p$, 
where $\nu(p,\lambda)$ is the rate parameter $\nu$ 
that produces edge density $p$ for the specific value of $\lambda$.
For $\lambda=0$ this relationship is simply $\nu = p/(1-p)$, but 
for $\lambda>0$
the relation is not so tidy, and in practice it is easier to fit $\nu$ numerically
rather than evaluate the expression.

When considering a collection of graph frequencies 
we can fit $\lambda$ by minimizing residuals with respect to the
model. Given a collection of $N$ graphs, let $y_k^1, \ldots, y_k^N$ 
be the vectors of $k$-node subgraph frequencies for each graph and
$p^1, \ldots, p^N$ be the edge densities. We can then 
fit $\lambda$ as:
$$
\lambda_k^{opt} = \arg \min_\lambda \sum_{i=1}^N 
|| \pi_k(\nu(p^i,\lambda),\lambda) -  y_k^i ||_2.
$$

\begin{figure}[t]
\begin{center}
\includegraphics[width=0.85\linewidth]{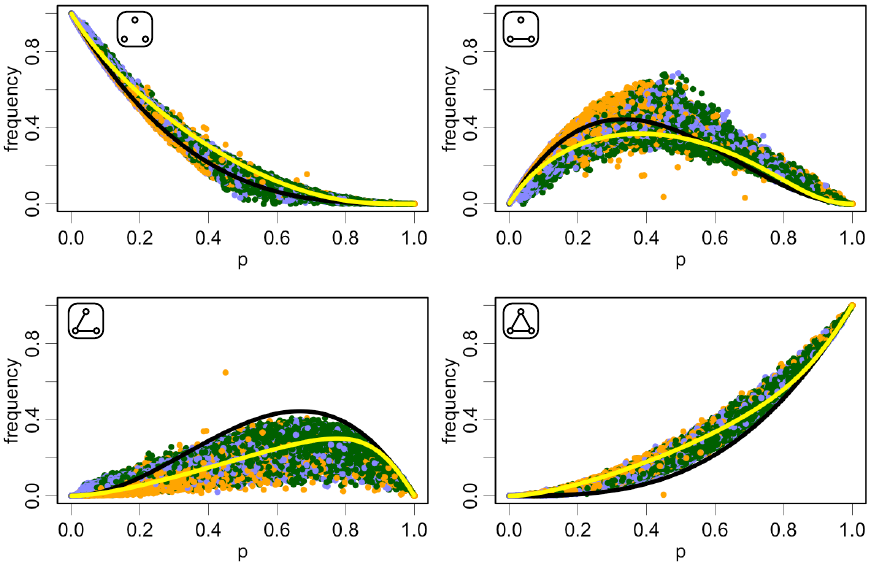} 
\vspace*{-0.2in}
\caption{Subgraph frequencies for 3-node subgraphs in 
50-node graphs, shown as a function of $p$.
The black curves illustrate $\Gnp$,
while the yellow curves illustrate the fit model. 
\vspace*{-0.2in}
}
\label{f:triadplambda}
\end{center}
\end{figure}

\begin{figure}[t]
\begin{center}
\includegraphics[width=0.75\columnwidth]{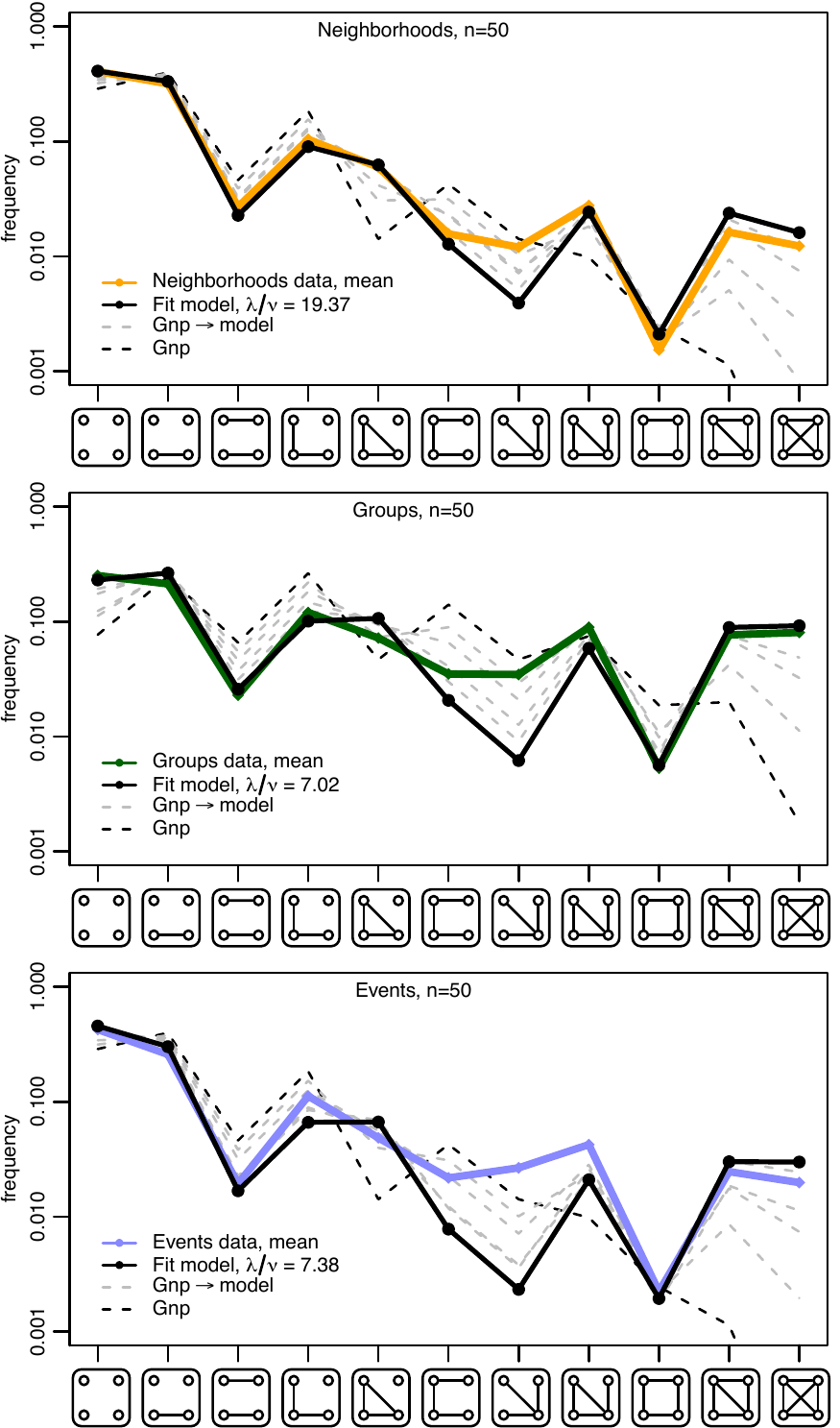} 
\vspace*{-0.1in}
\caption{
The four-node subgraph frequencies for the 
means of the 50-node graph collections in Figure \ref{f:triadplambda}, 
and the subgraph frequency of the model,
fitting the triadic closure rate $\lambda$ to the mean vectors.
As $\lambda$ increases from $\lambda=0$ to $\lambda=\lambda_{opt}$,
 we see how this single additional parameter provides a striking fit.
\vspace*{-0.2in}
}
\label{f:quadfit}
\end{center}
\end{figure}

In Figure \ref{f:triadplambda} we plot the three-node subgraph frequencies 
as a function of edge density $p$, for a collection of 300,000 50-node subgraphs, 
again a balanced mixture of neighborhoods, groups, and events. 
In this figure we also plot (in yellow) the curve resulting from fitting our random walk
model with triadic closure,
$\pi_k(\nu(p,\lambda^{opt}_k),\lambda^{opt}_k)$, which is thus parameterized
as a function of edge density $p$. For this mixture of collections and $k=3$,
the optimal fit is $\lambda^{opt}_3 =1.61$.
Notice how the yellow line deviates from the black
$\Gnp$ curve to better represent the backbone of natural graph frequencies. 
From the figure it is clear that almost all graphs have more triangles than
a sample from $\Gnp$ of corresponding edge density. When describing
extremal bounds in Section \ref{sec:bounds}, we will discuss how $\Gnp$
is in fact by no means the extremal lower bound. 

As suggested by Figure \ref{f:transitions}, examining the subgraph frequencies
for four-node subgraphs is fully tractable. In Figure \ref{f:quadfit}, we fit $\lambda$
to the mean subgraph frequencies of our three different collections of graphs separately.
Note that the mean of the subgraph frequencies over a set of graphs is not necessarily
itself a subgraph frequency corresponding to a graph, but we fit these mean 11-vectors
as a demonstration of the model's ability to fit an `average' graph.
The subgraph frequency of $\Gnp$ at the edge density corresponding to the
data is shown as a black dashed line in each plot --- with poor agreement --- 
and gray dashed lines illustrate an
incremental transition in $\lambda$, starting from zero (when
it corresponds to $\Gnp$) and ending at $\lambda^{opt}$.

The striking agreement between the fit model and the mean of each collection is 
achieved at the corresponding edge density by fitting only $\lambda$. For neighborhood
graphs, this agreement deviates measurably on only a single subgraph frequency, the
four-node star. The y-axis is plotted on a logarithmic scale, which makes it rather 
remarkable how precisely the model describes the scarcity of the four-node cycle. The
scarcity of squares has been previously observed in 
email neighborhoods on Facebook \cite{ugander2012}, 
and our model provides the first intuitive explanation of this scarcity.

The model's ability to characterize the backbone of the empirical graph frequencies
suggests that the subgraph frequencies of individual graphs can be usefully
studied as deviations from this backbone. 
In fact, we can interpret the fitting procedure for $\lambda$ 
as a variance minimization procedure. Recall that the mean of a set of
points in $\mathbb R^n$ is the point that minimizes the sum
 of squared residuals. In this way,  the procedure is in fact fitting the 
`mean curve' of the model distribution to the empirical subgraph frequencies.

Finally, our model can be used to provide
a measure of the triadic closure strength differentially
between graph collections, investigating the difference in $\lambda^{opt}$ 
for the subgraph frequencies of different graph collections. 
In Figure \ref{f:quadfit}, the three different graph types resulted
in notably different ratios of $\lambda/\nu$ --- the ratio of the
triadic closure formation rate to the basic process rate ---
with a significantly higher value for this ratio in neighborhoods.
We can interpret this as saying that open triads in neighborhoods are 
more prone to triadic closure than open triads in
groups or events.

\section{Extremal bounds}
\label{sec:bounds}

As discussed at the beginning of the previous section, 
we face two problems in analyzing the subgraph frequencies of real graphs:
to characterize the distribution of values we observe in practice, and
to understand the combinatorial structure of the overall space in which
these empirical subgraph frequencies lie.
Having developed stochastic models to address the former question, 
we now consider the latter question.

Specifically, in this section we characterize 
extremal bounds on the set of possible subgraph frequencies.
Using machinery from the theory of graph homomorphisms,
we identify fundamental bounds on the 
space of subgraph frequencies that are not properties of 
social graphs, but rather, are universal properties of all graphs.
By identifying these bounds, we make apparent large tracts
of the feasible region that are theoretically inhabitable but not populated by any of the
empirical social graphs we examine.

We first review a body of techniques
based in extremal graph theory and the theory of graph homomorphisms
\cite{lovasz}.
We use these techniques to formulate a set of inequalities on subgraph frequencies;
these inequalities 
are all linear for a fixed edge density, an observation that allows
us to cleanly construct a linear program to maximize and minimize 
each subgraph frequency within the combined constraints. In this manner,
we show how it is possible to map outer bounds on the geography of all 
these structural constraints. We conclude by offering two basic
propositions that transcend all edge densities, thus identifying fundamental 
limits on subgraph frequencies of all sizes.

\subsection{Background on subgraph frequency and homomorphism density}

In this subsection, we review some background arising from the theory of
graph homomorphisms.
We will use this homomorphism machinery to 
develop inequalities governing 
subgraph frequencies. These inequalities allow us to 
describe the outlines of the space underlying
Figure \ref{f:3d}(a) --- the first step in understanding
which aspects of the distribution of subgraph frequencies 
in the simplex are the result of empirical 
properties of human social networks, and
which are the consequences of purely combinatorial constraints.

\xhdr{Linear constraints on subgraph frequency}
Let $s(F,G)$ denote the subgraph frequency of $F$ in $G$, as
defined in the last section:
the probability that a random $|V(F)|$-node subset of $G$ induces
a copy of $F$.
Note that since $s(F,G)$ is a probability over outcomes, 
it is subject to 
the law of total probability. The law of total probability for subgraph 
frequencies takes the following form.

\begin{proposition}
\label{p:linear}
For any graph $F$ and any integer $\ell \ge k$, where $|V(F)|=k$, 
the subgraph density of $F$ in $G$, $s(F,G)$ satisfies the equality
$$
s(F,G) = \sum_{\{H:|V(H)|=\ell\}} s(F,H)s(H,G).
$$
\end{proposition}
\begin{proof}
Let $H'$ be a random $\ell$-vertex induced subgraph of $G$. 
Now, the set of outcomes $\mathcal H = \{H :|V(H)|=\ell\}$ form 
a partition of the sample space, each with probability $s(H,G)$. 
Furthermore, conditional upon an $\ell$-vertex induced subgraph 
being isomorphic to $H$, $s(F,H)$ is the probability that a random $k$-vertex 
induced subgraph of $H$ is isomorphic to $F$.
\end{proof}

This proposition characterizes an important property of 
subgraph frequencies: the vector of subgraph frequencies on $k$
nodes exists in a linear subspace of the vector of subgraph frequencies
on $\ell > k$ nodes. Furthermore, this means that any constraint
on the frequency of a subgraph $F$ will also constrain the frequency of 
any subgraph $H$ for which $s(F,H) >0$ or $s(H,F)>0$.

\xhdr{Graph homomorphisms}
A number of fundamental inequalities on the occurrence of subgraphs
are most naturally formulated in terms of {\em graph homomorphisms},
a notion that is connected to but distinct from the notion of induced subgraphs.
In order to describe this machinery, we first review some basic definitions
\cite{borgs}.
if $F$ and $G$ are labelled graphs, a map $f : V(F) \rightarrow V(G)$
is a {\em homomorphism} if each edge $(v,w)$ of $F$ maps to an edge
$(f(v),f(w))$ of $G$.
We now write $t(F,G)$ for the probability
that a random map from $V(F)$ into $V(G)$ is a homomorphism,
and we refer to $t(F,G)$ as a {\em homomorphism density} of $F$ and $G$.

There are three key differences between the homomorphism density
$t(F,G)$ and the subgraph frequency $s(F,G)$ defined earlier in this section.
First, $t(F,G)$ is based on mappings of $F$ into $G$ that can
be many-to-one --- multiple nodes of $F$ can map to the same node of
$G$ --- while $s(F,G)$ is based on one-to-one mappings.
Second, $t(F,G)$ is based on mappings of $F$ into $G$ that must
map edges to edges, but impose no condition on pairs of nodes in $F$
that do not form edges: in other words, a homomorphism is allowed
to map a pair of unlinked nodes in $F$ to an edge of $G$.
This is not the case for $s(F,G)$, which is based on maps that
require non-edges of $F$ to be mapped to non-edges of $G$.
Third, $t(F,G)$ is a frequency among mappings from labeled
graphs $F$ to labelled graphs $G$, while $s(F,G)$ is a frequency
among mappings from unlabeled $F$ to unlabeled $G$. 

From these three differences, it is not difficult to write down a
basic relationship governing the functions $s$ and $t$
\cite{borgs}.
To do this, it is useful to define the intermediate notion 
$\tinj(F,G)$, which is the probability that a random {\em one-to-one}
map from $V(F)$ to $V(G)$ is a homomorphism. 
Since only an $O(1/V(G))$ fraction of all maps from 
$V(F)$ to $V(G)$ are not one-to-one, we have
\begin{equation}
\label{eq:hominj1}
t(F,G) = \tinj(F,G) + O(1/|V(G)|).
\end{equation}
Next, by definition,
a one-to-one map $f$ of $F$ into $G$ is a homomorphism
if and only if the image $f(F)$, when viewed as an induced subgraph
of $G$, contains all of $F$'s edges and possibly others.
Correcting also for the conversion from labelled to unlabeled graphs, 
we have
\begin{equation}
\label{eq:hominj2}
\tinj(F,G) = \sum_{F':F \subseteq F'} \frac{\ext(F,F') \cdot \aut(F')}{k!} \cdot s(F',G),
\end{equation}
where $\aut(F')$ is the number of automorphisms of $F'$ and $\ext(F,F')$
is the number of ways that a labelled graph $F$ can be extended (by adding edges) to form a labelled
graph $H$ isomorphic to $F'$.

\xhdr{Homomorphism inequalities}
There are a number of non-trivial results bounding 
the graph homomorphism density, which we now review. 
By translating these to
the language of subgraph frequencies, we can begin to
develop bounds on the simplexes in
Figure~\ref{f:3d}.

For complete graphs, the Kruskal-Katona Theorem produces upper bounds
on homomorphism density in terms of the edge density while the Moon-Moser
Theorem provides lower bounds, also in terms of the edge density. 
\begin{proposition}[Kruskal-Katona \cite{lovasz}]
For a complete graph $K_r$ on $r$ nodes and graph $G$ with edge density $t(K_2,G)$,
$$t(K_r,G) \le t(K_2,G)^{r/2}.$$
\end{proposition}

\begin{proposition}[Moon-Moser \cite{MM62, Razborov08}]
For a complete graph $K_r$ on $r$ nodes and graph $G$ with edge density 
$t(K_2,G) \in [(k-2)/(k-1),1]$,
$$t(K_r,G) \ge \prod_{i=1}^{r-1} (1-i(1-t(K_2,G))).$$
\end{proposition}
The Moon-Moser bound is well known to not be sharp,
and Razborov has recently given an 
impressive sharp lower bound for the homomorphism density of the triangle
$K_3$ \cite{Razborov08} using sophisticated machinery \cite{Razborov07}. 
We limit our discussion to the simpler Moon-Moser lower bound which 
takes the form of a concise polynomial and provides bounds for arbitrary $r$, not just 
the triangle ($r=3$).

Finally, we employ a powerful inequality that
is known to lower bound the homomorphism 
density of any graph $F$ that is either a forest, an even cycle, or a complete
bipartite graph. Stated as such, it is the solved special cases of the open
 {\it Sidorenko Conjecture}, 
which posits that the result could be extended to all bipartite graphs $F$.
We will use the following proposition in particular when $F$ is a tree,
and will refer to this part of the result as the 
{\em Sidorenko tree bound}.

\begin{proposition}[Sidorenko \cite{lovasz,Sidorenko93}]
For a graph $F$ that is a forest, even cycle, or complete bipartite graph, 
with edge set E(F), and $G$ with edge density $t(K_2,G)$,
$$
t(F,G)
\ge t(K_2,G)^{|E(F)|}.$$
\end{proposition}

Using Equations (\ref{eq:hominj1}) and (\ref{eq:hominj2}),
we can translate statements about homomorphisms into asymptotic 
statements about the combined frequency of particular sets of subgraphs. 
We can also translate statements about frequencies of subgraphs
to frequencies of their complements using the following basic
fact.

\begin{lemma}
If for graphs $F_1, \ldots F_\ell$, coefficients $\alpha_i \in \mathbb R$, and a function $f$,
$$\alpha_1 s(F_1,G) + \ldots + \alpha_\ell s(F_\ell,G) \ge f(s(K_2,G)), \text{ } \forall G,$$
then 
$$\alpha_1 s(\comp F_1,G) + \ldots + \alpha_\ell s(\comp F_\ell, G) \ge f(1-s(K_2,G)), \text{ } \forall G.$$
\end{lemma}
\begin{proof}
Note that $s(F,G) = s(\comp F, \comp G)$. Thus if 
$$ \alpha_1 s(\comp F_1, \comp G) + \ldots + \alpha_\ell s(\comp F_\ell, \comp G) \ge f(s(\comp K_2,\comp G)), \text{ } \forall G,$$
then 
$$ \alpha_1 s(\comp F_1, G) + \ldots +  \alpha_\ell s(\comp F_\ell, G) \ge f(s(\comp K_2, G)), \text{ } \forall G,$$
where $s(\comp K_2, G)=1-s(K_2,G)$.
\end{proof}

\omt{
This lemma imples that the Kruskal-Katona and Moon-Moser bounds for 
complete graphs also furnish bounds for empty graphs, and each 
Sidorneko-type bound on a set of subgraph frequencies also 
furnishes bounds on the set of complementary graphs.
}

\begin{figure*}[t]
\begin{center}
\includegraphics[width=0.9\linewidth]{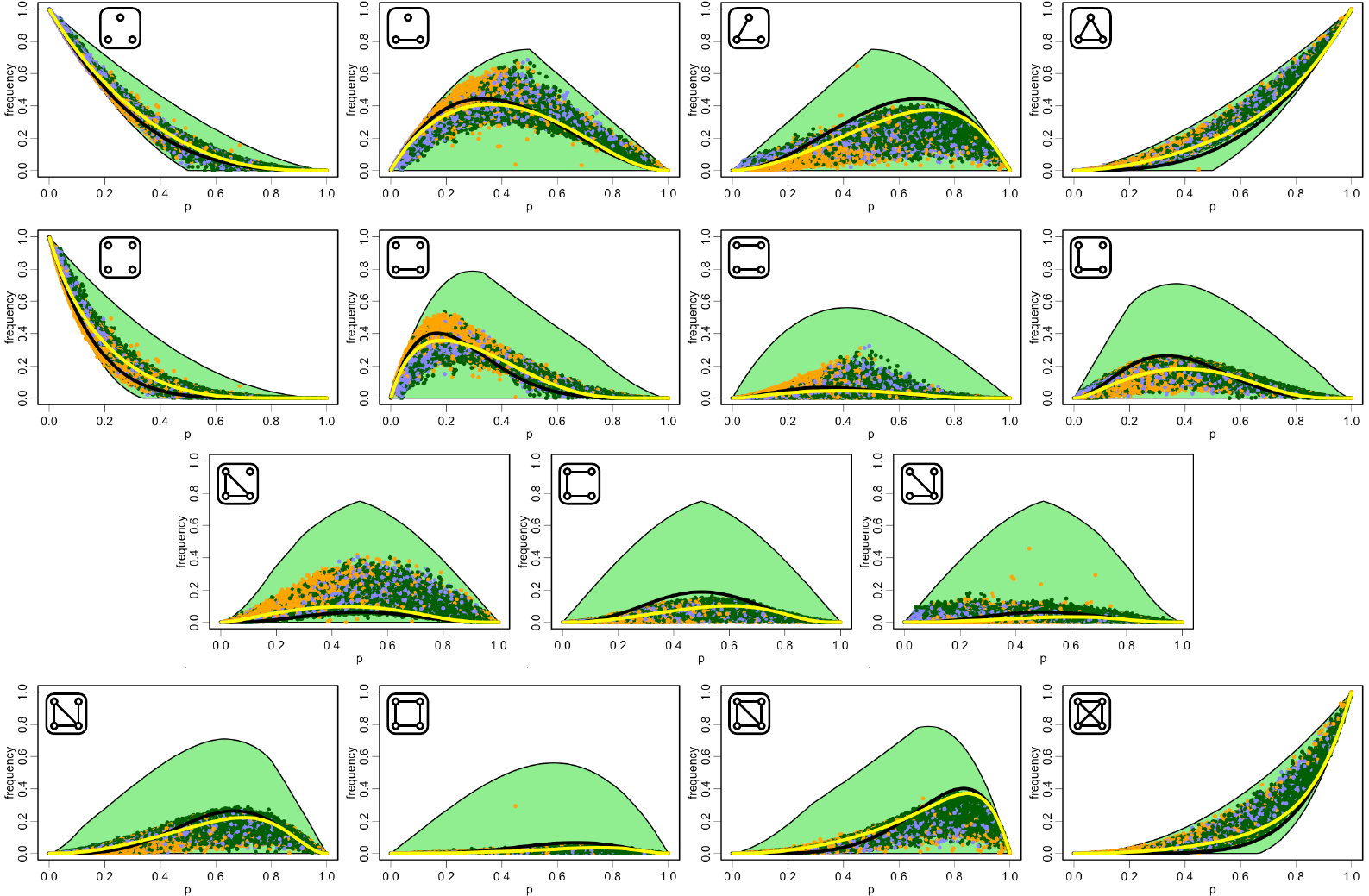} 
\caption{Subgraph frequencies for 3-node and 4-node subgraphs
 as function of edge density $p$. 
The light green regions denote the asymptotically feasible region
found via the linear program. The empirical frequencies are as in 
Figure \ref{f:triadplambda}. 
The black curves illustrate $\Gnp$,
while the yellow curves illustrate the fit triadic closure model. 
\vspace*{-0.2in}
}
\label{f:bounds}
\end{center}
\end{figure*}

\subsection{An LP for subgraph frequency bounds}

In the previous section, we reviewed linear constraints between 
the frequencies of subgraphs of different sizes, and upper and lower
bounds on graph homomorphism densities with applications to
subgraph frequencies. We will now use these constraints to assemble
a linear program capable to mapping out bounds on the extremal
geography of the subgraph space we are considering. To do this,
we will maximize and minimize the frequency of each individual
subgraph frequency, subject to the constraints we have just catalogued.

We will focus our analysis on the cases $k=3$, the triad frequencies,
and $k=4$, the quad frequencies. 
Let $x_1, x_2, x_3, x_4$ denote the subgraph frequencies 
$s(\cdot,G)$ of the four possible 3-vertex undirected graphs, 
ordered by increasing edge count.
\begin{program}
The frequency $x_i$ of a 3-node subgraph in any graph $G$ with 
edge density $p$ is bounded  asymptotically (in $|V(G)|$) by  
$ \max / \min x_i $ subject to $x_i \ge 0, \forall i$ and:
\begin{align}
x_1 + x_2 + x_3 + x_4 = 1, && \frac{1}{3} x_2 + \frac{2}{3} x_3 + x_4 = p, 
\label{p3:1} \\
x_4\le p^{3/2}, &&  x_1 \le (1-p)^{3/2}, 
\label{p3:2} \\
x_4\ge p(2p-1) && p \ge 1/2,
\label{p3:3} \\
x_1\ge (1-p)(1-2p) && p \le 1/2,
\label{p3:4} 
 \\
 (1/3) x_3 + x_4 \ge p^2, && x_1 + (1/3) x_2 \ge (1-p)^2.
 \label{p3:5}
\end{align}
\end{program}
Here the equalities in (\ref{p3:1}) derive from the linear constraints,
the constraints in (\ref{p3:2}) derive from Kruskal-Katona,
the constraints (\ref{p3:3}-\ref{p3:4}) derive from Moon-Moser,
and the constraints in (\ref{p3:5}) derive from the Sidorenko tree bound. 
More generally, we obtain the following general linear program
that can be used to find nontrivial bounds for any subgraph
frequency: 

\begin{program}
The frequency $f_F$ of a $k$-node subgraph $F$ in any graph $G$ with 
edge density $p$ is bounded  asymptotically (in $|V(G)|$) by  
$ \max / \min f_F$, 
subject to $Af_F=b(p), Cf_F\le d(p)$, appropriately assembled.
\end{program}

From Program 1 given above it is possible to derive a simple upper bound
on the frequency of the 3-node-path (sometimes described in the 
social networks literature as the ``forbidden triad'', as mentioned earlier).

\begin{proposition}
\label{forbidden}
The subgraph frequency of the 3-node-path $F$ obeys $s(F,G) \le 3/4 + o(1), \forall G$.
\end{proposition}
\begin{proof}
Let $x_1, x_2, x_3, x_4$ again denote the subgraph frequencies 
$s(\cdot,G)$ of the four possible 3-vertex undirected graphs, 
ordered by increasing edge count, where $x_3$ is the frequency
of the 3-node-path. By the linear constraints, 
$$(1/3) x_2 + (2/3) x_3 + x_4 = p,$$
while by Moon-Moser, $x_4 + O(1/|V(G)|) \ge p(2p-1)$. Combining
these two constraints we have:
 $$x_3 \le 3p(1-p) + o(1).$$
The polynomial in $p$ is maximized at $p=1/2$, giving an 
upper bound of $3/4 + o(1)$.
\end{proof}
This bound on the ``forbidden triad'' is immediately apparent 
from Figure \ref{f:bounds} as well,
which shows the bounds constructed via linear programs for all 
3-node and 4-node subgraph frequencies.  
In fact, the subgraph frequency of the `forbidden'' 3-node-path 
in the balanced complete bipartite graph $K_{n/2,n/2}$, 
which has edge density $p=1/2$, 
is exactly $s(F,G) = 3/4$, demonstrating that this bound is 
asymptotically tight.  (In fact, we can perform 
a more careful analysis showing
that it is exactly tight for even $n$.)

Figure \ref{f:bounds} illustrates these bounds for $k=3$ and $k=4$. 
Notice that our empirical distributions
of subgraph frequencies fall well within these bounds, leaving large tracts of the 
bounded area uninhabited by any observed dense social graph. While the
bounds do not fully characterize the feasible region of subgraph
frequencies, the fact that the bound is asymptotically tight at $p=1/2$ for the
complete bipartite graph $K_{n/2,n/2}$ is important --- practically no empirical social graphs 
come close to the boundary, despite this evidence that it is feasibly approachable. 
We emphasize that an exact characterization of the 
feasible space would necessitate machinery at least as sophisticated 
as that used by Razborov.

In the next subsection we develop two more general observations about the subgraph
frequencies of arbitrary graphs, the latter of which illustrates that, with the exception
of clique subgraphs and empty subgraphs, it is always possible to be free
from a subgraph. This shows that the lower regions of the non-clique non-empty
frequency bounds in
Figure \ref{f:bounds} are always inhabitable, despite the fact that social graphs do
not empirically populate these regions.

\subsection{\hspace{-0.4cm} Bounding frequencies of arbitrary subgraphs}

\def\eps{\varepsilon}
\def\inj{{\rm inj}}

The upper bound for the frequency of the 3-node-path given 
in Proposition \ref{forbidden} amounted to simply combining appropriate
upper bounds for different regions of possible edge densities $p$. 
In this section, we provide two general bounds pertaining to the 
subgraph frequency of an arbitrary subgraph $F$. First, we show 
that any subgraph that is not a clique and is not empty must have a 
subgraph density bounded strictly away from one. Second, we show
that for every subgraph $F$ that is not a clique and not empty, it is
always possible to construct a family of graphs with any specified
asymptotic edge density $p$ that contains no induced copies of $F$.

With regard to Figures \ref{f:bounds}, 
the first of the results in this subsection uses the
Sidorenko tree bound to show that
in fact no subgraph other than the clique or the empty graph,
not even for large values of $k$, has
a feasible region that can reach a frequency of $1 - o(1)$.
The second statement demonstrates that 
it is always possible to be free of any subgraph that is not
a clique or an empty graph, even if this does not occur in the
real social graphs we observe.

\begin{proposition}
For every $k$, there exist constants $\eps$ and $n_0$ such that
the following holds.
If $F$ is a $k$-node subgraph that is not a clique and not empty, and
$G$ is any graph on $n \geq n_0$ nodes, then 
$s(F,G) < 1 - \eps$.
\end{proposition}
\begin{proof}
Let $S_k$ denote the $k$-node star ---
in other words the tree consisting of a single node linked
to $k-1$ leaves.
By Equation (\ref{eq:hominj1}), if $G$ has $n$ nodes, then 
$\tinj(S_k,G) \geq t(S_k,G) - c/n$ for an absolute constant $c$.
We now state our condition on $\eps$ and $n_0$ in the statement
of the proposition: 
we choose $\eps$ small enough and $n_0$ large enough so that
\begin{equation}
\frac{(1 - \eps)^k}{2 {k \choose 2}^{k-1}} 
    > \max\left(\eps, \frac{c}{n}\right).
\label{eq:eps-condition}
\end{equation}

For a $k$-node graph $F$,
let ${\cal P}(F)$ denote the property that
for all graphs $G$ on at least $n_0$ nodes, we have
$s(F,G) < 1 - \eps$.
Our goal is to show that ${\cal P}(F)$ holds for all $k$-node $F$
that are neither the clique nor the empty graph.
We observe that since $s(F,G) = s(\overline{F},\overline{G})$,
the property ${\cal P}(F)$ holds if and only if 
${\cal P}(\overline{F})$ holds.

The basic idea of the proof is to consider any $k$-node graph $F$
that is neither complete nor empty, and to argue that the star $S_k$
lacks a one-to-one homomorphism into at least one of $F$ or
$\overline{F}$ --- suppose it is $F$.
The Sidorenko tree bound says that $S_k$ must have a non-trivial
number of one-to-one homomorphisms into $G$;
but the images of these homomorphisms must be places where 
$F$ is not found as an induced subgraph, and this puts an
upper bound on the frequency of $F$.

We now describe this argument in more detail; we start by
considering any specific $k$-node graph $F$ that is
neither a clique nor an empty graph.
We first claim that there cannot be a 
one-to-one homomorphism from $S_k$ into both of $F$ and $\overline{F}$.
For if there is a one-to-one homomorphism from $S_k$ into $F$,
then $F$ must contain a node of degree $k-1$; 
this node would then be isolated in $\overline{F}$, and hence
there would be no
one-to-one homomorphism from $S_k$ into $\overline{F}$.
Now, since it is enough to prove that just one of 
${\cal P}(F)$ or ${\cal P}(\overline{F})$ holds, we choose 
one of $F$ or $\overline{F}$ 
for which there is no one-to-one homomorphism from $S_k$.
Renaming if necessary, let us assume it is $F$.

Suppose by way of contradiction that $s(F,G) \geq 1 - \eps$.
Let $q$ denote the edge density of $F$ ---
that is, $q = |E(F)| / {k \choose 2}$.
The edge density $p$ of $G$ can be written, using
Proposition \ref{p:linear}, as 
\begin{eqnarray*}
p = s(K_2,G) & = &  \sum_{\{H:|V(H)|=k\}} s(K_2,H)s(H,G) \\
& \geq & s(K_2,F) s(F,G) \geq q (1 - \eps).
\end{eqnarray*}

By a {\em $k$-set} of $G$, we mean a set of $k$ nodes in $G$.
We color the $k$-sets of $G$ according to the following rule.
Let $U$ be a $k$-set of $G$:
we color $U$ {\em blue} if $G[U]$ is isomorphic to $F$, and 
we color $U$ {\em red} if there is a one-to-one homomorphism from
$S_k$ to $G[U]$.
We leave the $k$-set uncolored if it is neither blue nor red under these rules.
We observe that no $k$-set $U$ can be colored both blue and red,
for if it is blue, then $G[U]$ is isomorphic to $F$, and hence there is no
one-to-one homomorphism from $S_k$ into $G[U]$.
Also, note that $s(F,G) \geq 1 - \eps$ is equivalent to saying
that at least a $(1 - \eps)$ fraction of all $k$-sets are blue.

Finally, what fraction of $k$-sets are red?
By the Sidorenko tree bound, we have 
$$t(S_k,G) \geq p^{k-1} \geq q^k (1 - \eps)^k 
   \geq \frac{(1 - \eps)^k}{{k \choose 2}^{k-1}},$$
where the last inequality follows from the fact that 
$F$ is not the empty graph, and hence 
$q \geq 1 / {k \choose 2}$.
Since $\tinj(S_k,G) \geq t(S_k,G) - c/n$, our condition on $n$
from (\ref{eq:eps-condition}) implies that 
$$\tinj(S_k,G) \geq \frac{(1 - \eps)^k}{2 {k \choose 2}^{k-1}} > \eps.$$
Now, let $\inj(S_k,G)$ denote the number of one-to-one homomorphisms
of $S_k$ into $G$; by definition, 
$$\tinj(S_k,G) = \frac{\inj(S_k,G)}{n (n-1) \cdots (n-k+1)}
    = \frac{\inj(S_k,G)}{k! {n \choose k}},$$
and hence
$$\inj(S_k,G) = k! {n \choose k} \tinj(S_k,G) > \eps k! {n \choose k}.$$
Now, at most $k!$ different one-to-one homomorphisms can map $S_k$
to the same $k$-set of $G$, and hence more than 
$\eps {n \choose k}$ many $k$-sets of $G$ are red.
It follows that the fraction of $k$-sets that are red is
$> \eps$; but this contradicts our assumption that
at least a $(1 - \eps)$ fraction of $k$-sets are blue, since
no $k$-set can be both blue and red.
\end{proof}

\begin{proposition}
Assume $F$ is not a clique and not empty. Then for each edge
density $p$ there exists a sequence $G^p_1, G^p_2, \ldots$ of asymptotic
edge density $p$ for which 
$F$ does not appear as an induced subgraph in any $G^p_i$.
Equivalently, $s(F,G^p_i) = 0, \forall i$.
\end{proposition}
\begin{proof}
We call $H$ a {\em near-clique} if it has at most one
connected component of size greater than one, 
and this component is a clique.
For any $p \in [0,1]$, it is possible to construct an
infinite sequence $H^p_1, H^p_2, \ldots$ of
near-cliques with asymptotic density $p$, by simply 
taking the non-trivial component of each $H^p_i$ to be a
clique of the appropriate size.

Now, fix any $p \in [0,1]$, and let $F$ be any graph that is 
neither a clique nor an empty graph.
If $F$ is not a near-clique, then the required sequence
$G^p_1,G^p_2,\ldots$ is the sequence of near-cliques
$H^p_1,H^p_2,\ldots$, since all the induced 
subgraphs of a near-clique are themselves near-cliques.

On the other hand, 
if $F$ is a near-clique, then since $F$ is neither a clique nor
an empty graph, the complement of $F$ is not a near-clique. 
It follows that the required sequence $G^p_1,G^p_2,\ldots$
is the sequence of complements of the
near-cliques $H^{1-p}_1, H^{1-p}_2,\ldots$.
\end{proof}

Note that it is possible to take
an $F$-free graph with asymptotic density $p$ and append 
nodes with local edge density $p$ and random (Erd\H{o}s-R\'{e}nyi)
connections to obtain a graph with any intermediate subgraph frequency 
between zero and that of $\Gnp$. The same blending arguement can be
applied to any graph with a subgraph frequency above $\Gnp$ to again
find graphs with intermediate subgraph frequencies. In this way we
see that large tracts of the subgraph frequency simplex are fully feasible for
arbitrary graphs, yet by Figure \ref{f:bounds} are clearly not inhabited 
by any real world social graph.
 
\section{Classification of audiences}
\label{sec:classify}

The previous two sections characterize empirical 
and extremal properties of the space of subgraph frequencies, 
providing two complementary frameworks for understanding 
the structure of social graphs. 
In this section, we conclude our work with a 
demonstration of how subgraph frequencies can also provide
a useful tool for distinguishing between different categories of 
graphs. The Edge Formation Random Walk
model introduced in Section \ref{sec:triads} figures notably,
providing a meaningful
baseline for constructing classification features, contributing to the 
best overall classification accuracy we are able to produce.

Thus, concretely our classification task is to take a social graph and 
determine whether it is a node neighborhood, the set of people
in a group, or the set of people at an event.
This is a specific version of a broader characterization problem
that arises generally in social media --- namely how social audiences
differ in terms of social graph structure \cite{adamic2008}. Each of the
three graph types we discuss --- neighborhoods, groups, and events ---
define an audience with which a user may choose to converse. 
The defining feature of such audience
decisions has typically been their size --- as users choose to share
something online, do they want to share it publicly, with their friends, or
with a select subgroup of their friends? Products such as Facebook groups
exist in part to address this audience problem, enabling the creation of
small conversation circles. Our classification task is essentially asking:
do audiences differ in meaningful structural ways other than just size?

In Figure \ref{f:3d} and subsequently in Figure \ref{f:bounds}, we saw how
the three types of graphs that we study --- neighborhoods,
groups, and events --- are noticeably clustered around different 
structural foci in the space of subgraph frequencies. Figure \ref{f:bounds}
focused on graphs consisting of exactly 50-nodes, where it is visibly apparent
that both neighborhoods and events tend to have a lower edge density 
than groups of that size. Neighborhood edge density --- equivalent to
 the {\em local clustering coefficient} --- is known to generally decrease
 with graph size \cite{newman-networks-book,UKBM12},
but it is not clear that all three of the graph types we consider here 
should decrease at the same rate.

In Figure \ref{f:edgedens}, we see that in fact the three graph
types do not decrease uniformly, with the average edge density of
neighborhoods decreasing
more slowly than groups or events. Thus,
small groups are denser than neighborhoods while large groups are sparser, 
with the transition occurring at around 400 nodes. Similarly,
small event graphs are denser than neighborhoods while 
large events are much sparser, 
with the transition occurring already at around 75 nodes.

The two crossing points in Figure 6 suggest a curious challenge:
are their structural features of audience graphs that distinguish them
from each other even when they exhibit the same edge density?
Here we use the language of subgraph frequencies 
to formulate a classification task for 
classifying audience graphs based on subgraph frequencies. We compare
our classification accuracy to the accuracy achieved when also 
considering a generous vector of much more sophisticated graph features.
We approach this classification task using a simple logistic regression model.
While more advanced machine learning models capable of
learning richer relationships would likely produce 
better classification accuracies, 
our goal here is to establish that this vocabulary of
features based on subgraph frequencies can produce 
non-trivial classification results even in conjunction with
simple techniques such as logistic regression. 
Evaluating our features in other contexts such as  
graph matching \cite{Khan2011, Li2012, Vishwanathan2010}, where
frequencies of connected subgraphs have been used previously \cite{Shervashidze2009},  
would be interesting future work.

When considering neighborhood graphs, recall that we are not
including the ego of the neighborhoods as part of the graph,
while for groups and events the 
administrators as members of their graphs. 
As such, neighborhoods without their ego deviate
systematicallly from analogous audience graphs created as
groups or as events. In Figure \ref{f:edgedens} 
we also show the average edge density
of neighborhoods with their ego, adding one node and $n-1$ edges,
noting that the difference is small for larger graphs.

\begin{figure}[t]
\begin{center}
\vspace*{-0.15in}
\includegraphics[width=0.95\columnwidth]{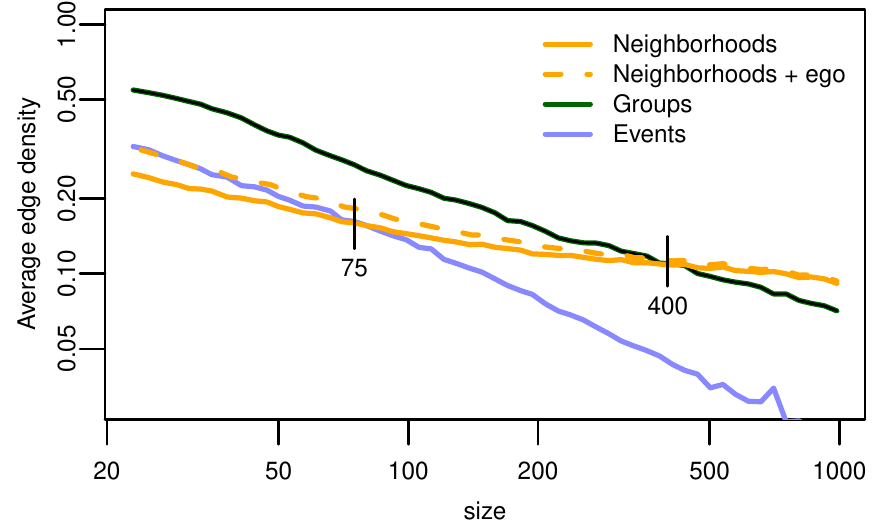} 
\vspace*{-0.15in}
\caption{Edge densities of neighborhoods, groups, and events as a function 
of size, $n$. 
When $n<400$, groups are denser then neighborhoods. 
When $n<75$, events are denser then neighborhoods.
}
\vspace*{-0.25in}
\label{f:edgedens}
\end{center}
\end{figure}

\xhdr{Classification features}
Subgraph frequencies has been the motivating
coordinate system for the present work, and will serve as our main
feature set. Employing the Edge Formation Random Walk model from
Section \ref{sec:triads}, we additionally describe each graph by its residuals
with respect to a backbone --- described by the parameter $\lambda$ --- 
fit to the complete unclassified training set. 

Features based on subgraph frequencies are local features, 
computable by examining only a few local nodes of the graph at a time. 
Note that the subgraph frequencies of arbitrarily large graphs can be
accurately approximated by sampling a small number of induced graphs.
Comparatively, it is relevant to ask: can these simple local features
do as well as more sophisticated {\em global} graph features? Perhaps
the number of connected components, the size of the largest component,
or other global features provide highly informative features for graph 
classification. 

To answer this question, we compare our classification 
accuracy using subgraph frequencies with the accuracy
we are able to achieve using a set of global
graph features. We consider:
\begin{itemize}
\addtolength{\itemsep}{-0.6\baselineskip}  
\item Size of the $k$ largest components, for $k=1,2$.
\item Size of the $k$-core, for $k=0,1,2,3$.
\item Number of components in the $k$-core, for $k=0,1,2$.
\item Degeneracy, the largest $k$ for which the $k$-core is non-empty. 
\item Size of the $k$-brace \cite{ugander2012}, for $k=1,2,3$.
\item Number of components in the $k$-brace, for $k=1,2,3$.
\end{itemize}

These features combine linearly to produce a rich set of
graph properties. For example, the number of components in the 
$1$-core minus the number of components in the $0$-core yields
the number of singletons in the graph. 

\xhdr{Classification results}
The results of the classification model are shown in Table \ref{t:acc},
reported in terms of classification accuracy --- the fraction of correct
classifications on the test data -- measured using
five-fold cross-validation on a balanced set of 10,000 instances. 
The classification tasks were chosen to be thwart
classification based solely on edge density, which indeed performs poorly. 
Using only 4-node subgraph frequencies and residuals, 
an accuracy of $77\%$ is achieved in both tasks.

In comparison, classification based on a set of global graph features
performed worse, achieving just $69\%$ and $76\%$ accuracy for the
two tasks.  Meanwhile, combining global and subgraph frequency
features performed best of all, with a classification accuracy of $81{-}82\%$.
In each case we also report the accuracy with and without residuals as features.
Incorporating residuals with respect to either a $\Gnp$ or 
Edge Formation Random Walk baseline consistently improved classification, and
examining residuals with respect to either baseline clearly provides a useful
orientation of the subgraph coordinate system for empirical graphs. 

\begin{table}
\centering
\vspace*{-0.15in}
\begin{tabular}{l||c|c}
Model Features & N vs. E, $n=75$ & N vs. G, $n=400$ \\
\hline 
\hline
Edges  & 0.487& 0.482 \\
\hline
Triads & 0.719& 0.647    \\
Triads + $R_G$ & 0.737& 0.673    \\
Triads + $R_\lambda$ & 0.736 & 0.668    \\
\hline
Quads & 0.751& 0.755  \\
Quads + $R_G$ & 0.765 & 0.769  \\
Quads + $R_\lambda$ & 0.765 & 0.769  \\
\hline
Global + Edges & 0.694& 0.763  \\
Global + Triads  & 0.785 & 0.766   \\
Global + Triads + $R_G$  & 0.784 & 0.766 \\
Global + Triads + $R_\lambda$ & 0.789 & 0.767  \\
Global + Quads & 0.797  & 0.812 \\
Global + Quads + $R_G$ & 0.807 & 0.815 \\
Global + Quads + $R_\lambda$ & 0.809 & 0.820 \\
\end{tabular}
\vspace*{-0.15in}
\caption{
Classification accuracy for N(eighborhoods), 
G(roups), and E(vents) on different sets of features.
$R_G$ and $R_\lambda$ denote the residuals with 
respect to a $\Gnp$ and stochastic graph model baseline,
as described in the text.
}
\vspace*{-0.15in}
\label{t:acc}

\end{table}

\section{Conclusion}

The modern study of social graphs has primarily focused
on the examination of the sparse large-scale structure of
human relationships. This global perspective has led to fruitful
theoretical frameworks for the study of many networked domains, notably
the world wide web, computer networks,
and biological ecosystems \cite{newman-networks-book}.
However, in this work we argue
that the locally dense structure of social graphs admit an additional 
framework for analyzing the structure of social graphs.

In this work, we examine the structure of social graphs through the
coordinate system of subgraph frequencies, developing two
complementary frameworks
 that allow us to identify both `social' structure and 
`graph' structure. The framework developed in Section \ref{sec:triads} 
enables us to characterize the apparent social forces guiding graph formation,
while the framework developed in Section \ref{sec:bounds} characterizes
fundamental limits of all graphs, delivered through combinatorial 
constraints.
Our coordinate system and frameworks are not only useful for developing intuition,
but we also demonstrate how they can be used to accurately classify graph types
using only these simple descriptions in terms of subgraph frequency. 

\xhdr{Distribution note}
Implementations
of the Edge Formation Random Walk equilibrium solver 
and the subgraph frequency bounds optimization program
are available from the first author's webpage.

\vspace*{-0.1in}
\xhdr{Acknowledgments}
We thank Peter Grassberger for helpful comments. This work was supported in part by NSF grants IIS-0910664 and IIS-1016099.
\vfill

\bibliographystyle{plain}


\begin{thebibliography}{10}

\bibitem{adamic2008}
L.~Adamic, J.~Zhang, E.~Bakshy, and M.~S. Ackerman.
\newblock Knowledge sharing and Yahoo Answers: everyone knows something.
\newblock In {\em WWW}, pages 665--674. ACM, 2008.

\bibitem{BBKLR11}
L.~Backstrom, E.~Bakshy, J.~Kleinberg, T.~Lento, and I.~Rosenn.
\newblock Center of attention: how Facebook users allocate attention across
  friends.
\newblock In {\em ICWSM}, 2011.

\bibitem{bollobas-rand-graphs-book}
B.~Bollob\'as.
\newblock {\em Random Graphs}.
\newblock Cambridge University Press, second edition, 2001.

\bibitem{borgs}
C.~Borgs, J.~T. Chayes, L.~Lov\'asz, V.~Sos, B.~Szegedy, and K.~Vesztergombi.
\newblock Counting graph homomorphisms.
\newblock In {\em Topics in Discrete Mathematics,} eds. M.~Klazar, J.~Kratochvil, M.~Loebl, J.~Matousek, R.~Thomas, and
  P.~Valtr, pages 315--371.
  Springer, 2006.
  
\bibitem{chatterjee2011}
S.~Chatterjee and P.~Diaconis.
\newblock Estimating and understanding exponential random graph models.
\newblock {\em arXiv preprint}, arXiv:1102.2650v3, 2011.

\bibitem{Davis1967}
J.~Davis and S.~Leinhardt. 
\newblock The structure of positive interpersonal relations in small groups.
\newblock In {\em Sociological Theories in Progress. Vol. 2}, eds. 
J.~Berger, M.~Zelditch, and B.~Anderson. Houghton-Mifflin, 1971.

\bibitem{faust2007}
K.~Faust.
\newblock Very local structure in social networks.
\newblock {\em Sociological Methodology}, 37(1):209--256, 2007.

\bibitem{faust2010}
K.~Faust.
\newblock A puzzle concerning triads in social networks: Graph constraints and
  the triad census.
\newblock {\em Social Networks}, 32(3):221--233, 2010.

\bibitem{welser}
D.~Fisher, M.~A. Smith, and H.~T. Welser.
\newblock You are who you talk to: Detecting roles in usenet newsgroups.
\newblock In {\em HICSS}, 2006.

\bibitem{granovetter-weak-ties}
M.~Granovetter.
\newblock The strength of weak ties.
\newblock {\em American Journal of Sociology}, 78:1360--1380, 1973.

\bibitem{Inokuchi2000}
A.~Inokuchi, T.~Washio, and H.~Motoda.
\newblock An apriori-based algorithm for mining frequent substructures from graph data.
\newblock In {\em PKDD '00}, pages 13--23, 2000.

\bibitem{jin-girvan-newman}
E.~M. Jin, M.~Girvan, and M.~E.~J. Newman.
\newblock The structure of growing social networks.
\newblock {\em Phys. Rev. E}, 64:046132, 2001.

\bibitem{kuramochi2004}
M.~Kuramochi and G.~Karypis. 
\newblock An efficient algorithm for discovering frequent subgraphs.
\newblock {\em IEEE Trans. on Knowledge and Data Engineering}, 16(9), 2004.

\bibitem{Khan2011}
A.~Khan, N.~Li, X.~Yan, Z.~Guan, S.~Chakraborty, and S.~Tao.
\newblock Neighborhood based fast graph search in large networks.
\newblock In {\em SIGMOD}, pages 901--912, 2011.


\bibitem{leskovec2005}
J.~Leskovec, J.~Kleinberg, and C.~Faloutsos.
\newblock Graphs over time: densification laws, shrinking diameters and
  possible explanations.
\newblock In {\em KDD}, pages 177--187. ACM, 2005.

\vfill\eject

\bibitem{Li2012}
G.~Li, M.~Semerci, B.~Yener, and M.~Zaki.
\newblock Effective graph classification based on topological and label attributes.
\newblock {\em Statistical Analysis and Data Mining}, 4(5):265--283, 2012.

\bibitem{lovasz}
L.~Lov\'asz.
\newblock Very large graphs.
\newblock In D.~Jerison, B.~Mazur, T.~Mrowka, W.~Schmid, R.~Stanley, and S.~T.
  Yau, editors, {\em Current Developments in Mathematics}, pages 67--128.
  International Press, 2009.

\bibitem{Milo2002}
R.~Milo, S.~Shen-Orr, S.~Itzkovitz, N.~Kashtan, D.~Chklovskii, and U.~Alon.
\newblock Network motifs: simple building blocks of complex networks.
\newblock {\em Science}, 298:824--827, 2002.

\bibitem{MM62}
J.~W. Moon and L.~Moser.
\newblock On a problem of {T}uran.
\newblock {\em Magyar Tud. Akad. Mat. Kutat Int. Kzl}, 7:283--286, 1962.

\bibitem{newman-networks-book}
M.~E.~J. Newman.
\newblock {\em Networks: An Introduction}.
\newblock Oxford University Press, 2010.

\bibitem{park2005}
J.~Park and M.~E.~J. Newman.
\newblock Solution for the properties of a clustered network.
\newblock {\em Phys Rev E}, 72(2):26136, 2005.

\bibitem{rapoport-triadic}
A.~Rapoport.
\newblock Spread of information through a population with socio-structural bias
  {I}: {A}ssumption of transitivity.
\newblock {\em Bulletin of Mathematical Biophysics}, 15(4):523--533, December
  1953.

\bibitem{Razborov07}
A.~Razborov.
\newblock Flag algebras.
\newblock {\em Journal of Symbolic Logic}, 72:1239--1282, 2007.

\bibitem{Razborov08}
A.~Razborov.
\newblock On the minimal density of triangles in graphs.
\newblock {\em Combinatorics, Probability and Computing}, 17:603--618, 2008.

\bibitem{Shervashidze2009}
N.~Shervashidze, S.~Vishwanathan, T.~Petri, K.~Mehlhorn, and K.~Borgwardt.
\newblock Efficient graphlet kernels for large graph comparison.
\newblock In {\em AISTATS}, 2009.



\bibitem{Sidorenko93}
A.~Sidorenko.
\newblock A correlation inequality for bipartite graphs.
\newblock {\em Graphs and Combinatorics}, 9:201--204, 1993.

\bibitem{strauss1986}
D.~Strauss.
\newblock On a general class of models for interaction.
\newblock {\em SIAM Review,} 28(4):513--527, 1986.

\bibitem{ugander2012}
J.~Ugander, L.~Backstrom, C.~Marlow, and J.~Kleinberg.
\newblock Structural diversity in social contagion.
\newblock {\em PNAS}, 109(16):5962--5966, 2012.

\bibitem{UKBM12}
J.~Ugander, B.~Karrer, L.~Backstrom, and C.~Marlow.
\newblock The anatomy of the Facebook social graph.
\newblock Technical Report cs.SI/1111.4503, arxiv, November 2011.

\bibitem{Vishwanathan2010}
S.~Vishwanathan, N.~Schraudolph, R.~Kondor, and K.~Borgwardt.
\newblock Graph kernels.
\newblock {\em Journal of Machine Learning Research}, 99:1201--1242, 2010.


\bibitem{voss-wikipedia}
J.~Voss.
\newblock Measuring {W}ikipedia.
\newblock In {\em ICISSI}, 2005.

\bibitem{wasserman-faust}
S.~Wasserman and K.~Faust.
\newblock {\em Social Network Analysis: Methods and Applications}.
\newblock Cambridge Univ. Press, 1994.

\bibitem{yan2002}
X.~Yan and J.~Han. 
\newblock gpsan: Graph-based substructure pattern mining. 
\newblock In {\em ICDM '02}, pages 721--724, 2002.

\end{thebibliography}

\end{document}